\documentclass[a4paper]{jpconf}
\usepackage{graphicx}
\usepackage{amssymb,color}
\usepackage{epsfig}
\usepackage{epstopdf}
\usepackage{epsfig}
\usepackage{amsmath}
\usepackage{amsthm}

\def\cb{{\mathcal B}}

\def\bb{{\mathbb B}}

\def\bi{{\mathbb I}}

\def\a{\alpha}
\def\b{\beta}
  \def\G{\Gamma}

\def\e{\eta}
\def\l{\lambda} 

\def\m{\mu}

\def\s{\sigma} \def\S{\Sigma}

\def\v{\varphi} \def\F{\Phi}

\def\w{\omega} \def\Om{\Omega}
\def\ab{\overline{a}} \def\bb{\overline{b}} \def\cb{\overline{c}}
\def\L{\ell}

\def\h{{\mathbf{h}}}

\def\x{\theta_1} \def\y{\theta_2} \def\z{\theta_3}
\def \ub{{\mathbf{u}}}
\def\tb{{\mathbf{t}}}

\newtheorem{thm}{Theoram}[section]
\newtheorem{lem}[thm]{Lemma}

\newtheorem{rem}{Remark}[section]
\newtheorem{ex}{Example}[section]

\begin{document}
\title{On Ground States and Phase Transitions of $\lambda$-Model on the Cayley Tree}

\author{Farrukh Mukhamedov$^1$, Chin Hee Pah$^2$, Hakim Jamil$^3$.}

\address{$^1$Department of Mathematical Sciences, College of Science, The United Arab Emirates
University, P.O. Box, 15551, Al Ain, Abu Dhabi, UAE}

\address{$^{2,3}$Department of Computational \& Theoretical Sciences, Faculty of Science, International Islamic University Malaysia P.O. Box, 141, 25710, Kuantan}

\ead{email: farrukh.m@uaeu.ac.ae}

\begin{abstract}
In the paper, we consider the $\l$-model with spin values $\{1, 2,
3\}$ on the Cayley tree of order two. We first describe ground
states of the model. Moreover, we also proved the existence of
translation-invariant Gibb measures for the $\l$-model which yield
the existence of the phase transition. Lastly, we established the
exitance of 2-periodic Gibbs measures for the model.
\end{abstract}

\section{Introduction}

The choice of Hamiltonian for concrete systems of interacting
particles represents an important problem of equilibrium statistical
mechanics \cite{Ba}. The matter is that, in considering concrete
real systems with many (in abstraction, infinitely many) degrees of
freedom, it is impossible to account for all properties of such a
system without exceptions. The main problem consists in accounting
only for the most important features of the system, consciously
removing the other particularities. On the other hand, the main
purpose of equilibrium statistical mechanics consists in describing
all limit Gibbs distributions corresponding to a given Hamiltonian
\cite{Ge}. This problem is completely solved only in some
comparatively simple cases. In particular, if there are only binary
interactions in the system then  problem of describing the limit
Gibbs distributions simplifies. One of the important models in
statistical mechanics is Potts models. These models describe a
special and easily defined class of statistical mechanics models.
Nevertheless, they are richly structured enough to illustrate almost
every conceivable nuance of the subject. In particular, they are at
the center of the most recent explosion of interest generated by the
confluence of conformal field theory, percolation theory, knot
theory, quantum groups and integrable systems \cite{Ma,NS}. The
Potts model \cite{Po} was introduced as a  generalization of the
Ising model to more than two components. At present the Potts model
encompasses a number of problems in statistical physics (see, e.g.
\cite{W}). Investigations of phase transitions of spin models on
hierarchical lattices showed that they make the exact calculation of
various physical quantities \cite{DGM},\cite{P1,P2},\cite{T}. Such
studies on the hierarchical lattices begun with development of the
Migdal-Kadanoff renormalization group method where the lattices
emerged as approximants of the ordinary crystal ones. In
\cite{PLM1,PLM2} the phase diagrams of the $q$-state Potts models on
the Bethe lattices were studied and the pure phases of the the
ferromagnetic Potts model were found. In \cite{G,GMM06} using those
results, uncountable number of the pure phase of the 3-state Potts
model were constructed. These investigations were based on a
measure-theoretic approach developed in \cite{Pr},\cite{PLM1,PLM2}.
The structure of the Gibbs measures of the Potts models has been
investigated in \cite{G,GR}.

It is natural to consider a model which is more complicated than
Potts one, in \cite{M} we proposed to study  so-called $\l$-model on
the Cayley tree (see also \cite{R,Roz}). In the mentioned paper, for
special kind of $\l$-model, its disordered phase is studied (see
\cite{GR,MR3}) and some its algebraic properties are investigated.
In the present, we consider symmetric $\l$-model with spin values
$\{1, 2, 3\}$ on the Cayley tree of order two. This model is much
more general than Potts model, and exhibits interesting structure of
ground states. We first describe ground states of the model.
Moreover, we also proved the existence of translation-invariant Gibb
measures for the $\l$-model which yield the existence of the phase
transition. Lastly, we established the exitance of 2-periodic Gibbs
measures for the model.

\section{Preliminaries}

Let $\Gamma^k_+ = (V,L)$ be a semi-infinite Cayley tree of order
$k\geq 1$ with the root $x^0$ (whose each vertex has exactly $k+1$
edges, except for the root $x^0$, which has $k$ edges). Here $V$ is
the set of vertices and $L$ is the set of edges. The vertices $x$
and $y$ are called {\it nearest neighbors} and they are denoted by
$l=\langle{x,y}\rangle$ if there exists an edge connecting them. A collection of
the pairs $\langle{x,x_1}\rangle,\dots,\langle{x_{d-1},y}\rangle$ is called a {\it path} from
the point $x$ to the point $y$. The distance $d(x,y), x,y\in V$, on
the Cayley tree, is the length of the shortest path from $x$ to $y$.
$$
W_{n}=\left\{ x\in V\mid d(x,x^{0})=n\right\}, \ \
V_n=\bigcup\limits_{m=1}^{n}W_m , \ \
L_{n}=\left\{
l=<x,y>\in L\mid x,y\in V_{n}\right\}.
$$
The set of direct successors of $x$ is defined by
$$
S(x)=\left\{ y\in W_{n+1}:d(x,y)=1\right\}, x\in W_{n}.
$$
Observe that any vertex $x\neq x^{0}$ has $k$ direct successors and
$x^{0}$ has $k+1$.

Now we are going to introduce a coordinate structure in $\G_+^k$.
Every vertex $x$ (except for $x^{0}$) of $\G_+^k$ has coordinates
$(i_1,\dots,i_n)$, here $i_m\in\{1,\dots,k\},\ 1\leq m\leq n$ and
for the vertex $x^0$ we put $(0)$ (see Figure 1). Namely, the
symbol $(0)$ constitutes level $0$ and the sites $i_1,\dots,i_n$
form level $n$ of the lattice. In this notation for $x\in\G_+^k,\
x=\{i_1,\dots,i_n\}$ we have
$
S(x)=\{(x,i): 1\leq i\leq k\},
$
here $(x,i)$ means that $(i_1,\dots,i_n,i)$.

\begin{figure}
    \begin{center}
        \includegraphics[width=10.07cm]{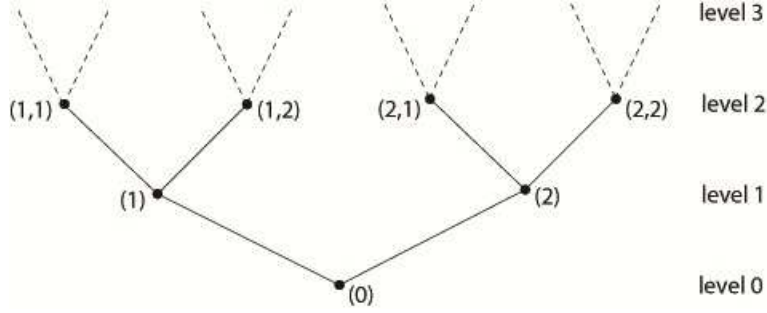}
    \end{center}
    \caption{The first levels of $\G_+^2$} \label{fig1}
\end{figure}

Let us define on $\G_+^k$ a binary operation $\circ:\G^k_+\times\G_+^k\to\G_+^k$
as follows, for any two elements $x=(i_1,\dots,i_n)$ and $y=(j_1,\dots,j_m)$ put
$$
x\circ y=(i_1,\dots,i_n)\circ(j_1,\dots,j_m)=(i_1,\dots,i_n,j_1,\dots,j_m)
$$
and
$$
y\circ x=(j_1,\dots,j_m)\circ(i_1,\dots,i_n)=(j_1,\dots,j_m,i_1,\dots,i_n).
$$
By means of the defined operation $\G_+^k$ becomes a
noncommutative semigroup with a unit. Using this semigroup
structure one defines translations $\tau_g:\G_+^k\to\G_+^k,\
g\in\G_k$ by
$$
\tau_g(x)=g\circ x.
$$

Let $G\subset\G_+^k$ be a sub-semigroup of $\G_+^k$ and $h:V\to \mathbb{R}$ be a function. We say that $h$ is a $G$-{\it periodic} if $h(\tau_g(x))=h(x)$ for all $x\in V$,$g\in G$ and $l\in L$.
Any $\G^k_+$-periodic function is called {\it translation-invariant}. Put
$$
G_m=\left\{x\in\G_+^k: d(x,x^0)\equiv0(\mathrm{mod }\ m)\right\},\ \ \ m\geq2.
$$

One can check that $G_m$ is a sub-semigroup with a unit.

Let us consider some examples. Let m=2, k=2, then $G_2$ can be
written as follows:
\[G_2=\left\{(0),(i_1,i_2,\dots,i_{2n}),n \in \mathbb{N}\right\}\]
In this case, $G_2$-periodic function $h$ has the following form:

\begin{equation}
    \h(x)=\left\{
    \begin{array}{ll}
        h_1&,x=(i_1,i_2,\dots,i_{2n}),\\
        h_2&,x=(i_1,i_2,\dots,i_{2n+1})
    \end{array}\right.
\end{equation}
for $i_k \in \{1,2\}$ and $k \in V$.
\begin{figure}[h!]
    \begin{center}
        \includegraphics[width=13cm]{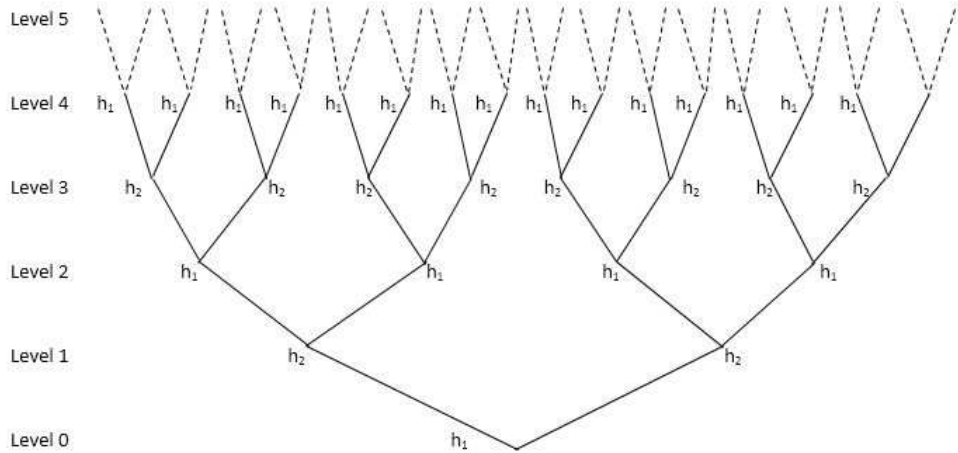}
    \end{center}
    \caption{Cayley tree for $G_2$} \label{example}
\end{figure}

In this paper, we consider the models where the spin takes values in
the set $\F=\{1,2,\dots,q\}$ and is assigned to the vertices of the
tree. A configuration $\s$ on V is then defined as a function $x \in
V \rightarrow \s(x)\in \F$; the set of all configurations coincides
with $\Om=\F^{\G^{k}}$. The Hamiltonian the $\l$-model has the
following form
\begin{eqnarray}\label{ham}
    H(\sigma)=\sum\limits_{<x,y>\in L}\l(\sigma (x),\sigma (y))
\end{eqnarray}
where the sum is taken over all pairs of nearest-neighbor vertices
$\langle{x,y}\rangle$, $\s \in \Om$. From a physical point of view
the interactions between particles do not depend on their locations,
therefore from now on we will assume that $\l$ is a symmetric
function, i.e. $\l(u,v)=\l(v,u)$ for all $u,v\in \mathbb{R}$.

We note that $\l$-model of this type can be considered as
generalization of the Potts model. The Potts model corresponds to
the choice  $\l(x,y)=-J\delta_{xy}$, where $x,y,J \in \mathbb{R}$.

In what follows, we restrict ourself to the case $k=2$ and $\F=\{1,2,3\}$, and for the sake of simplicity, we consider the following function:
\begin{equation}\label{cond}
    \l(i,j)=\left\{
    \begin{array}{lll}
        \ab,& \ \textrm{if} \ \ &|i-j|=2,\\
        \bb,&\ \textrm{if} \ \ &|i-j|=1,\\
        \cb,&\ \textrm{if} \ \ &i=j,
    \end{array}\right.
\end{equation}
where $\ab,\bb,\cb\in \mathbb{R}$ for some given numbers.

\begin{rem}
    We point out the considered model is more general then well-known Potts model \cite{W}, since if $\ab=\bb=0,\cb \neq 0$,
    then this model reduces to the mentioned model.
\end{rem}

\section{Ground States}

In this section, we describe ground state of the $\l$-model on a
Cayley tree. For a pair of configurations $\s$ and $\v$ coinciding
almost everywhere, i.e., everywhere except finitely many points, we
consider the relative Hamiltonian $H(\s,\v)$ determining the energy
differences of the configurations $\s$ and $\v$:

\begin{eqnarray}\label{eq12}
H(\s,\v)=\sum_{\substack{<x,y>\\ x,y \in V}}(\l(\s(x),\s(y))-\l(\v(x),\v(y)))
\end{eqnarray}

For each $x \in V$, the set $\{x,S(x)\}$ is called {\it a ball}, and
it is denoted by $b_x$. The set of all balls we denote by $M$.

We define the energy of the configuration $\s_b$ on b as follows
\[U(\s_b)=\dfrac{1}{2}\sum_{\substack{<x,y>\\ x,y \in V}}(\l(\s(x),\s(y)))\]
From \eqref{eq12}, we got the following lemma.

\begin{lem}
    The relative Hamiltonian \eqref{eq12} has the form

    \[H(\s,\v)=\sum_{b \in M}(U(\s_b)-U(\v_b)).\]
\end{lem}

\begin{lem}
    The inclusion
    \begin{eqnarray}\label{eq13}
    U(\v_b)\in \left\{\dfrac{\a + \b}{2}: \ \forall \a,\b \in \{\ab,\bb,\cb\}\right\}
    \end{eqnarray}
    holds for every configuration $\v_b$ on $b$ ($b \in M)$.
\end{lem}

   A configuration $\v$ is called a \textit{ground state} of the relative Hamiltonian $H$ if
    \begin{eqnarray}\label{eq14}
    U(\v_b)= \min \left\{\dfrac{\a + \b}{2}: \ \forall \a,\b \in \{\ab,\bb,\cb\}\right\}
    \end{eqnarray}
    for any $b \in M$

For any configuration $\s_b$, we have
\[U(\s_b)\in\{U_1,U_2,U_3,U_4,U_5,U_6\},\]
where
\begin{eqnarray}\label{U_n}
U_1=\ab&,U_2=(\ab+\bb)/2&,U_3=(\ab+\cb)/2,\nonumber\\
U_4=\bb&,U_5=(\bb+\cb)/2&,U_6=\cb.
\end{eqnarray}
We denote
\begin{eqnarray}\label{eq15}
A_m=\big\{(\ab,\bb,\cb)\in \mathbb{R}^3|\  U_m=\min_{1 \leq k \leq
6}\{U_k\}\big\}
\end{eqnarray}
Using \eqref{eq15}, we obtain
$$\begin{array}{llllll}
A_1=\left\{(\ab,\bb,\cb)\in \mathbb{R}^3| \ \ab \leq \bb,\ab \leq
\cb\right\},\ \ A_2=\left\{(\ab,\bb,\cb)\in \mathbb{R}^3| \ \ab = \bb\leq \cb\right\},\\
A_3=\left\{(\ab,\bb,\cb)\in \mathbb{R}^3| \ \ab = \cb,\leq
\bb\right\},\ \ \
A_4=\left\{(\ab,\bb,\cb)\in \mathbb{R}^3| \ \bb \leq \ab,\bb \leq \cb\right\},\\
A_5=\left\{(\ab,\bb,\cb)\in \mathbb{R}^3| \ \bb = \cb\leq
\ab\right\},\ \ \  \
A_6=\left\{(\ab,\bb,\cb)\in \mathbb{R}^3| \ \cb \leq \ab,\cb \leq \bb\right\}.\\
\end{array}$$

Now, we want to find ground states for each considered cases. To do
so, we introduce some notation. For each sequence
$\{k_0,k_1,\dots,k_n,\dots\}, k_n\in \{1,2,3\} $, $n \in
\mathbb{N}\cup \{0\}$, we define a configuration $\s$ on $\Om$ by
\begin{eqnarray*}
    \s(x)=k_{\L}, \ \ \textrm{if}\ \  x \in W_{\L}&,\L \geq 0.
\end{eqnarray*}
This configuration is denoted by $\s_{[k_n]}$.\begin{figure}[h!]
    \begin{center}
        \includegraphics[width=13cm]{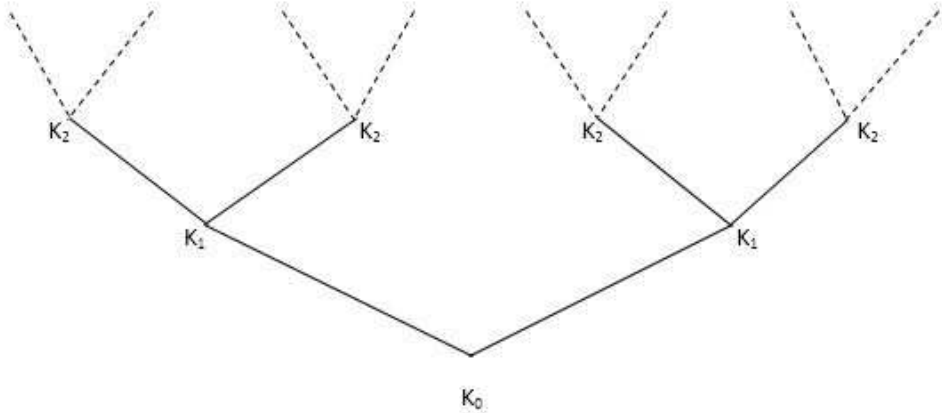}
    \end{center}
    \caption{Configuration for $\s_{[k_n]}$} \label{k0}
\end{figure}

If the sequence $\{k_0,k_1,\dots,k_n,\dots\}$ is $n$-periodic,(i.e. $k_{\L+n}=k_{\L},\forall n \in \mathbb{N}$), then instead of $\{k_0,k_1,\dots,k_n,\dots\}$, we write $\{k_0,k_1,\dots,k_{n-1}\}$. Correspondingly, the associated configuration is denoted by $\s_{[k_0,k_1,\dots,k_{n-1}]}$
\begin{thm}
    Let $(\ab,\bb,\cb) \in A_1$, then there are only two $G_2$-periodic ground states.
\end{thm}

\begin{proof}
    Let $(\ab,\bb,\cb) \in A_1$, then one can see that for this triple, the minimal value is $\ab$, which is achieved by the configuration on b, given in Figure \ref{fig2}.
    \begin{figure}[h!]
        \begin{center}
            \includegraphics[width=7cm]{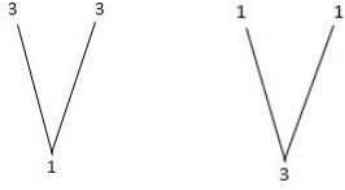}
        \end{center}
        \caption{Configurations for $A_1$} \label{fig2}
    \end{figure}

    Now using Figure \ref{fig2}, for each $n\in \mathbb{N}$, one can construct configurations on $\Om$ defined by:
    \begin{eqnarray*}
        \s^{(2)}_1=\s_{[1,3]},&\s^{(2)}_2=\s_{[3,1]}.
    \end{eqnarray*}
    \begin{figure}[h!]
        \begin{center}
            \includegraphics[width=11cm]{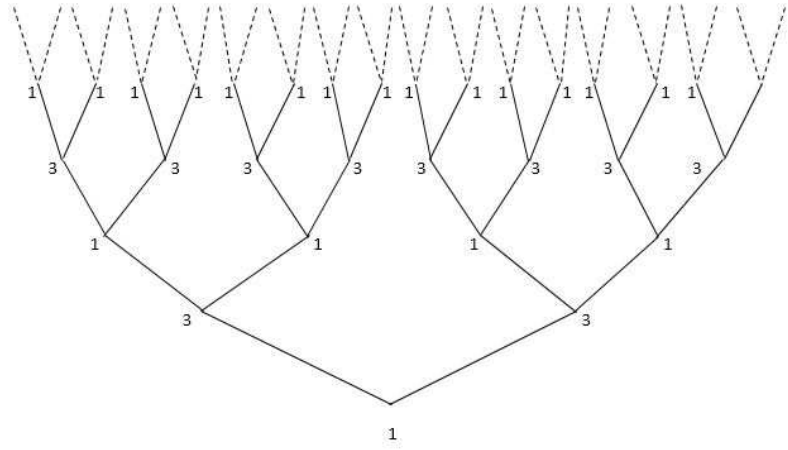}
        \end{center}
        \caption{Configuration for $\s^{(2)}_1$} \label{1313}
    \end{figure}

    Then, we can see that for any $b \in M$, one has
    \[U(\s^{(2)}_{b_{1,2}})=\min_{1 \leq k \leq 6}\{U_k\}\]
    which means $\s^{(2)}_{1,2}$ is a ground state. Moreover, $\s^{(2)}_{1,2}$ is $G_2$-periodic. Note that all ground states will coincide with these ones.
\end{proof}

\begin{thm}
    Let $(\ab,\bb,\cb) \in A_2$, then the following statements hold:
    \begin{itemize}
        \item[(i)] for every $n \in \mathbb{N}$, there is $G_n$-periodic ground state;
        \item[(ii)] there is uncountable number of ground states.
    \end{itemize}
\end{thm}

\begin{proof}
Let $(\ab,\bb,\cb) \in A_2$, then one can see that for this triple,
the minimal value is $(\ab+\bb)/2$, which is achieved by the
configurations on $b$ given in Figure \ref{fig3}.
    \begin{figure}[h!]
        \begin{center}
            \includegraphics[width=9cm]{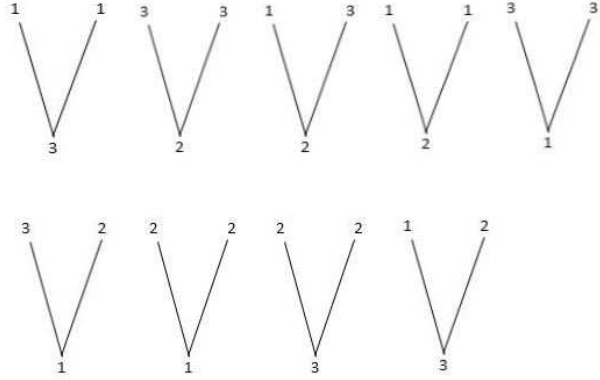}
        \end{center}
        \caption{Configurations for $A_2$} \label{fig3}
    \end{figure}
(i) Now using Figure \ref{fig3}, for each $n \in \mathbb{N}$, one
can construct configurations on $\Om$ defined by
        \begin{eqnarray*}
            \s^{(2n)}=\s_{\underbrace{[1,(2,3),\dots,(2,3),2]}_\text{2n}},\\\s^{(2n+1)}=\s_{\underbrace{[1,(2,3),\dots,(2,3)]}_\text{2n+1}}.
        \end{eqnarray*}
        \begin{figure}[h!]
            \begin{center}
                \includegraphics[width=12cm]{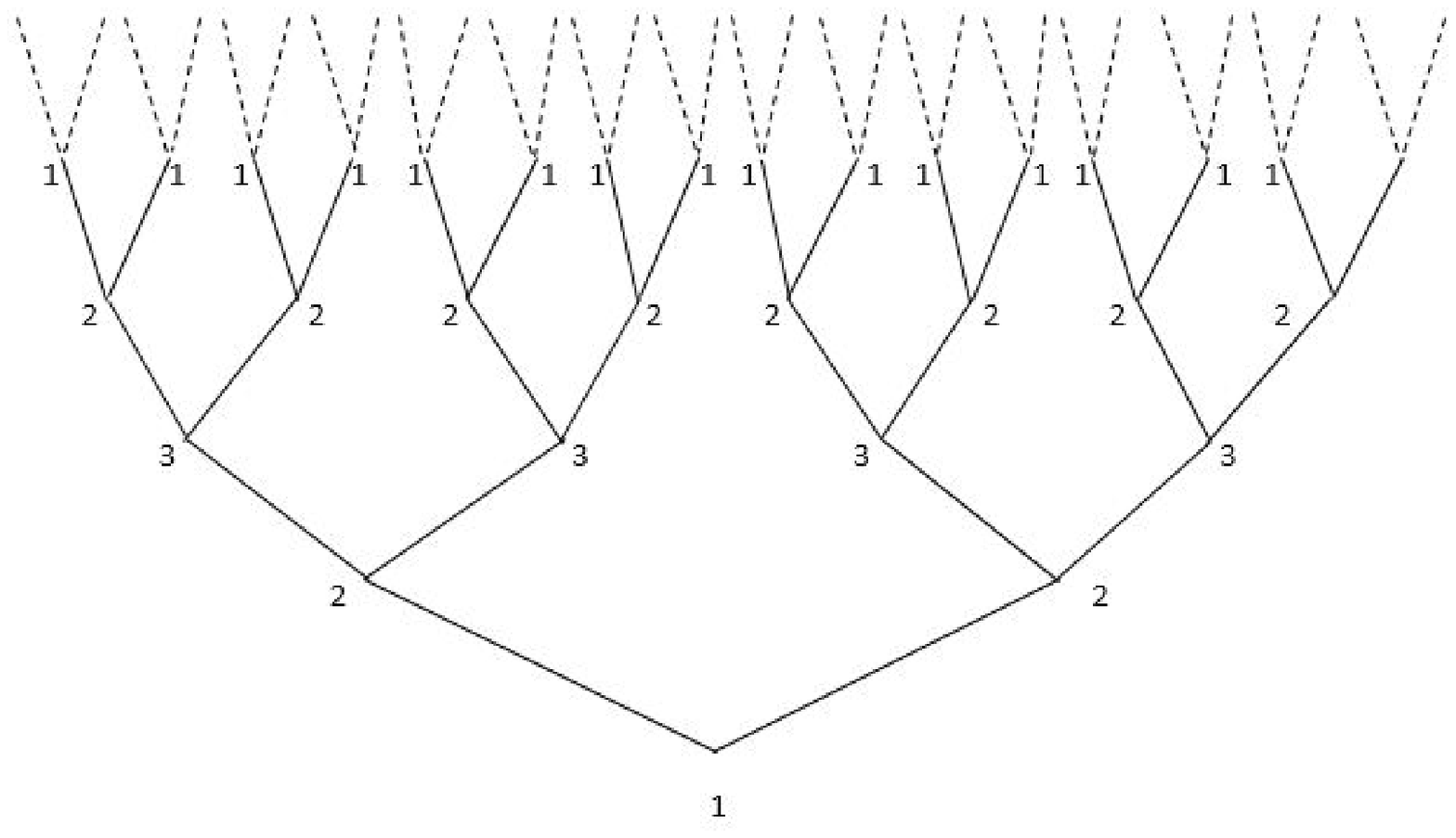}
            \end{center}
            \caption{Configuration for $\s^{(2)}$} \label{1232}
        \end{figure}

        Then, we can see that for any $b \in M$, one has
        \[U(\s^{(\xi)})=\min_{1 \leq k \leq 6}\{U_k\},\ \ \xi \in \{2n,2n+1\}\]
        which means $\s^{(n)}$ is a $G_n$-periodic ground state.

  (ii) To construct uncountable number of ground states, we consider the set
        \begin{eqnarray}\label{ncount}
        \S_{1,2,3}=\left\{(\tb_n)| \tb_n \in \{1,2,3\},\delta_{\{t_n,t_{n+1}\}}=0, n \in \mathbb{N}\right\}
        \end{eqnarray}
        where $\delta$ is the Kroneker delta. One can see that the
        set $\S_{1,2,3}$ is uncountable. Take any $\tb=(t_n)\in \S_{1,2,3}$. Let us construct a configuration by
        \begin{eqnarray*}
            \s^{(\tb)}=
            \left\{
            \begin{array}{ll}
                1,& x=(0), \\
                \tb_k,& x \in W_k,k \in V.
            \end{array}
            \right.
        \end{eqnarray*}
        \begin{figure}[h!]
            \begin{center}
                \includegraphics[width=13cm]{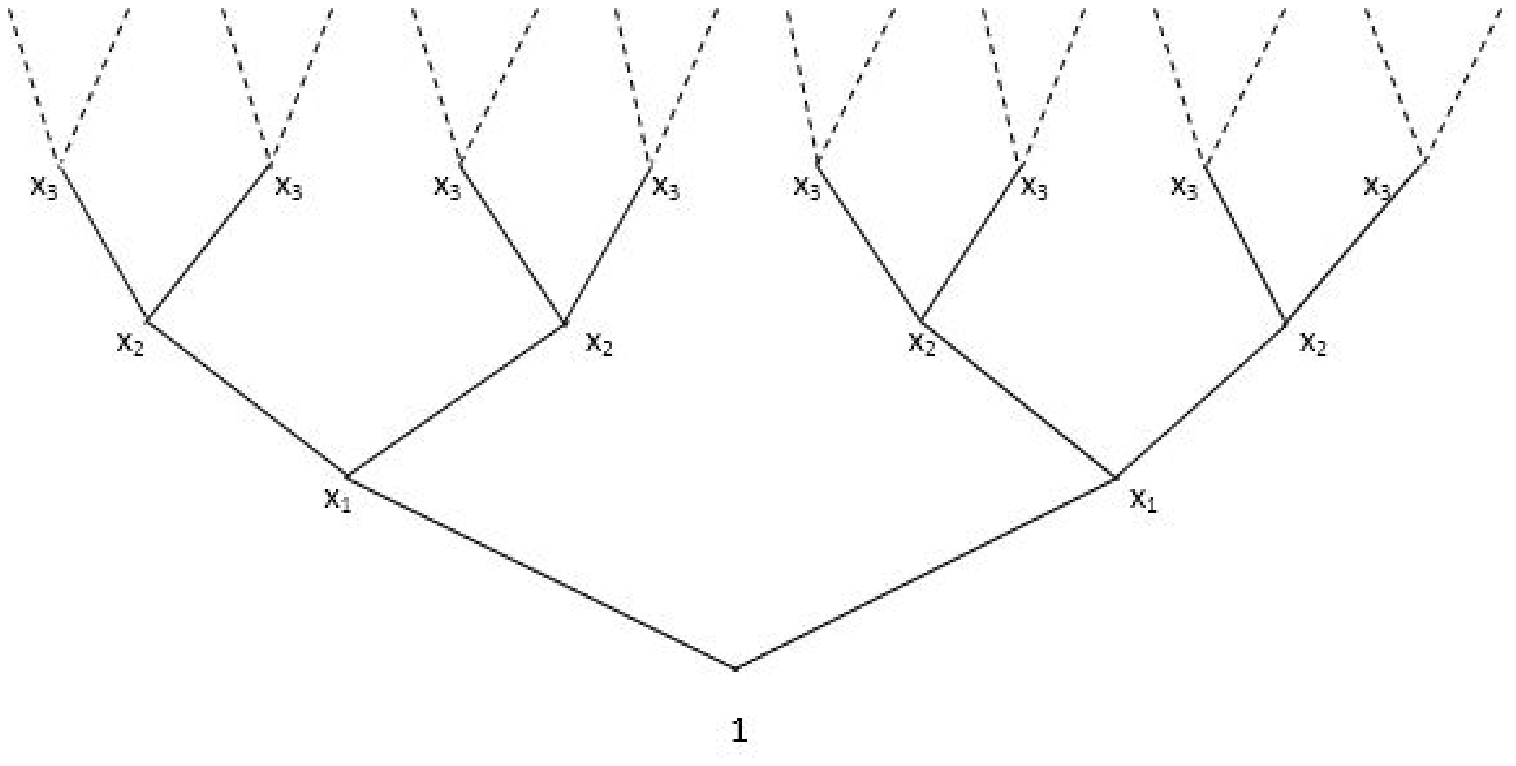}
            \end{center}
            \caption{Configuration for $\s^{(\tb)}$} \label{x1}
        \end{figure}
One can check that $\s^{(t)}$ is a ground state, and the
correspondence $\tb \in \S_{1,2,3} \rightarrow \s^{(\tb)}$ shows
that the set $\{\s^{(\tb)},\tb\in \S_{1,2,3}\}$ is uncountable. This
completes the proof.
\end{proof}

\begin{thm}
    Let $(\ab,\bb,\cb) \in A_3$, then the following statements hold:
    \begin{itemize}
        \item[(i)] there are three translation-invariant ground states;
        \item[(ii)] for every $n \in \mathbb{N}$, there is $G_n$-periodic ground state.
    \end{itemize}
\end{thm}
\begin{proof}
    Let $(\ab,\bb,\cb) \in A_3$, then one can see that for this triple, the minimal value is $(\ab+\cb)/2$, which is achieved by the configurations on b given in Figure \ref{fig4}.
    \begin{figure}[h!]
        \begin{center}
            \includegraphics[width=8cm]{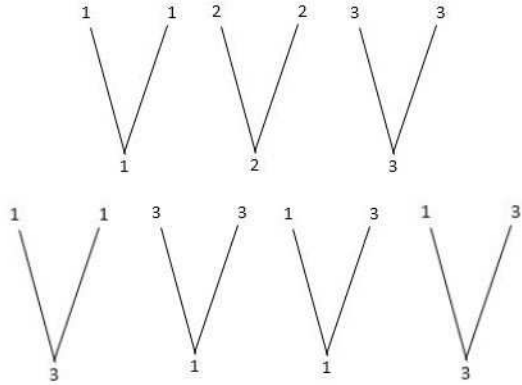}
        \end{center}
        \caption{Configuration for $A_3$} \label{fig4}
    \end{figure}

(i) In this case, we have three
        \[\s^{(k)}=\s_{[k]} ,k=\{1,2,3\} \] configurations, which are translation-invariant ground states.

  (ii) Using Figure \ref{fig4}, for each $n \in \mathbb{N}$, one can construct a configuration on $\Om$ defined by
        \begin{eqnarray*}
            \s^{(n)}=\s_{\underbrace{[1,3,3,\dots,3]}_\text{n}}.
        \end{eqnarray*}
        \begin{figure}[h!]
            \begin{center}
                \includegraphics[width=12cm]{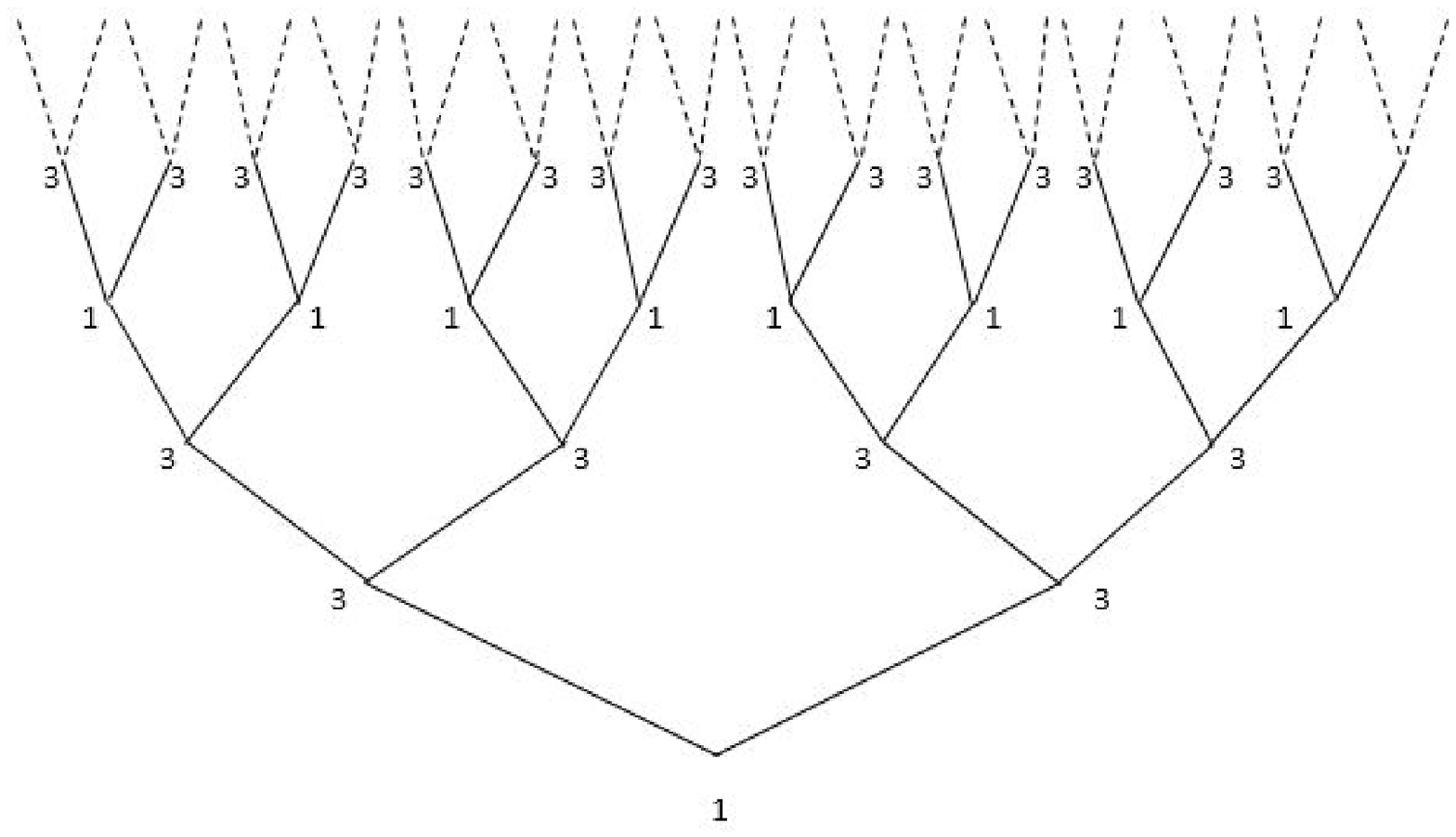}
            \end{center}
            \caption{Configuration for $\s^{(n)}$} \label{1331}
        \end{figure}

        Then, we can see that for any $b \in M$, one has
        \[U(\s^{(n)})=\min_{1 \leq k \leq 6}\{U_k\}\]
        which means $\s^{(n)}$ is a $G_{n}$-periodic ground state.
\end{proof}

\begin{thm}
    Let $(\ab,\bb,\cb) \in A_4$, then for every $n \in \mathbb{N}$, there is $G_{(3n+1)}$-periodic ground state.
\end{thm}

\begin{proof}
    Let $(\ab,\bb,\cb) \in A_4$, then one can see that for this triple, the minimal value is $\bb$, which is achieved by the configurations on b given, in Figure \ref{fig5}.
    \begin{figure}[h!]
        \begin{center}
            \includegraphics[width=7cm]{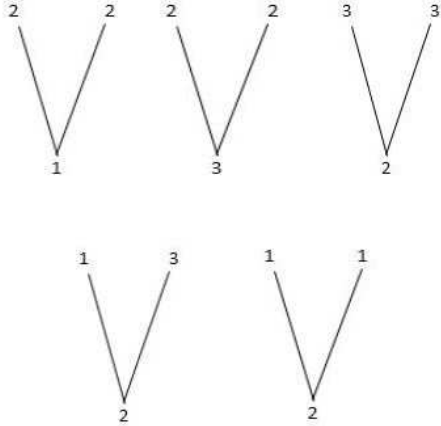}
        \end{center}
        \caption{Configurations for $A_4$} \label{fig5}
    \end{figure}

    Now, using Figure \ref{fig5}, for each $n \in \mathbb{N}$, one can construct a configuration on $\Om$ defined by
    \begin{eqnarray*}
        \s^{(3n+1)}=\s_{[1,\underbrace{(2,3,2),\dots,(2,3,2)}_n]}
    \end{eqnarray*}
    \begin{figure}[h!]
        \begin{center}
            \includegraphics[width=9cm]{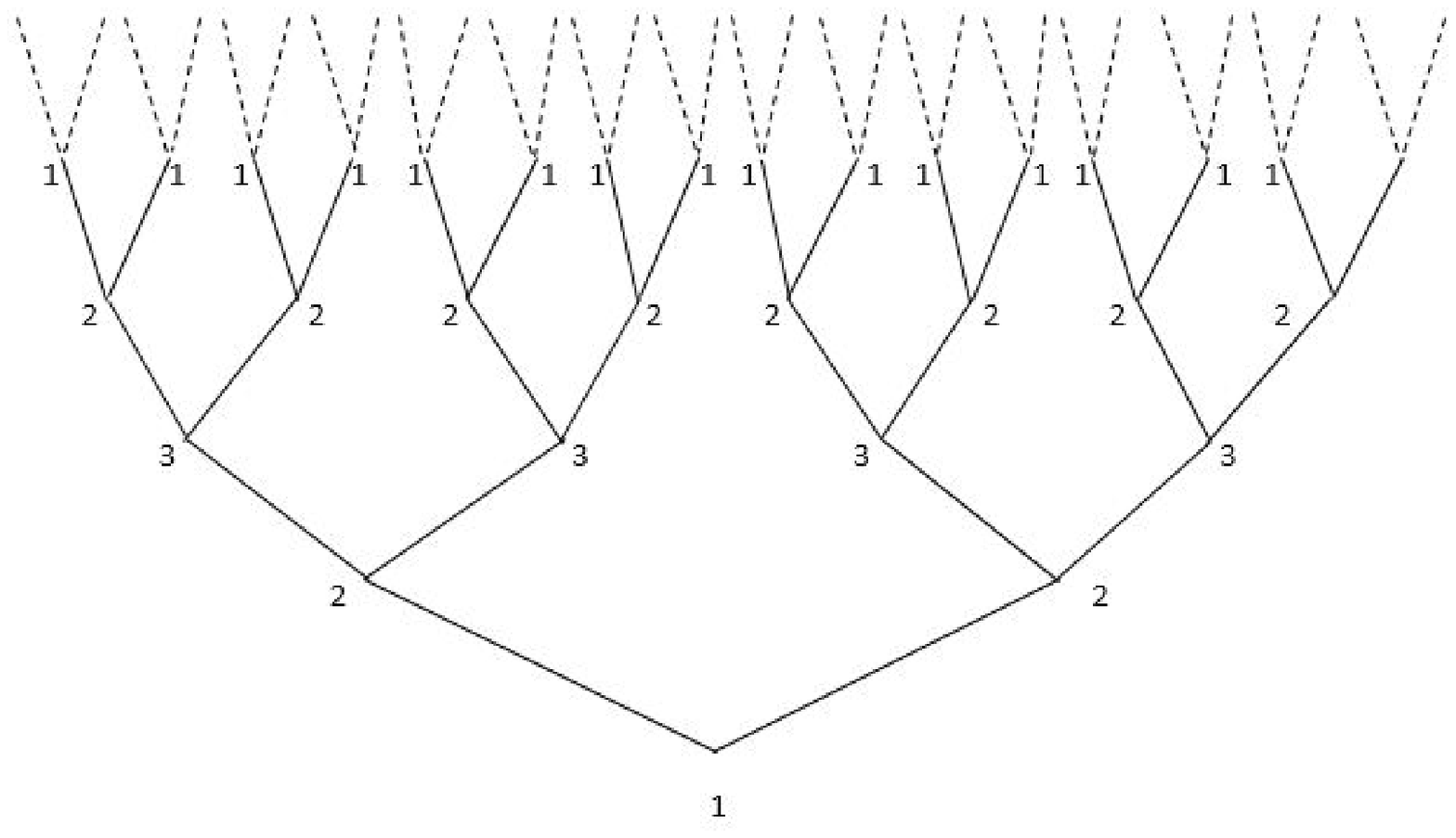}
        \end{center}
        \caption{Configuration for $\s^{(4)}$}
    \end{figure}
    Then, we can see that for any $b \in M$, one has
    \[U(\s^{(3n+1)})=\min_{1 \leq k \leq 6}\{U_k\}\]
    which means $\s^{(3n+1)}$ is a $G_{(3n+1)}$-periodic ground state.
\end{proof}

\begin{thm}
    Let $(\ab,\bb,\cb) \in A_5$, then the following statements hold:
    \begin{itemize}
        \item[(i)] there are three translation-invariant ground states;
        \item[(ii)] for every $n \in \mathbb{N}$,there is $G_n$-periodic ground state;
        \item[(iii)] there is uncountable number of ground states.
    \end{itemize}
\end{thm}
\begin{proof}
    Let $(\ab,\bb,\cb) \in A_5$, then one can see that for this triple, the minimal
    value is $(\bb+\cb)/2$, which is achieved by the configurations on b given in Figure \ref{fig6}.
    \begin{figure}[h!]
        \begin{center}
            \includegraphics[width=8cm]{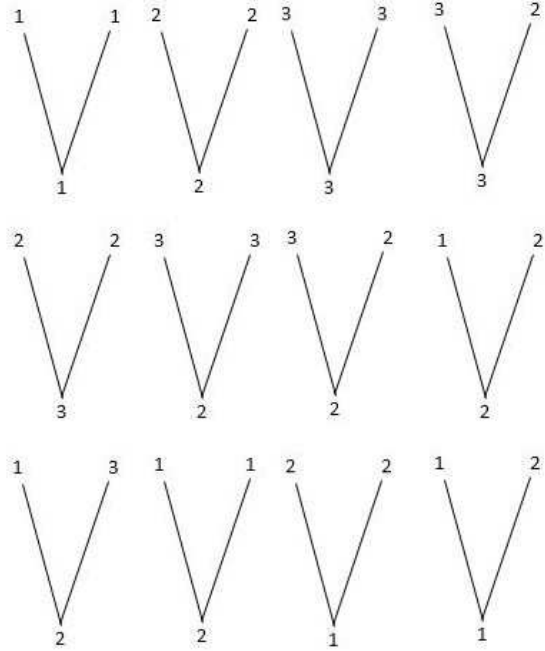}
        \end{center}
        \caption{Configuration for $A_5$} \label{fig6}
    \end{figure}

(i) In this case, we have three configurations:
        \[\s^{(k)}=\s_{[k]} ,k=\{1,2,3\} \]  which are translation-invariant ground states.

 (ii) Now, using Figure \ref{fig6}, for each $n\in \mathbb{N}$, one can see a configuration on $\Om$ defined by
        \begin{eqnarray*}
            \s^{(n)}=\s_{\underbrace{[1,2,2,\dots,2]}_\text{n}}.
        \end{eqnarray*}
        \begin{figure}[h!]
            \begin{center}
                \includegraphics[width=13cm]{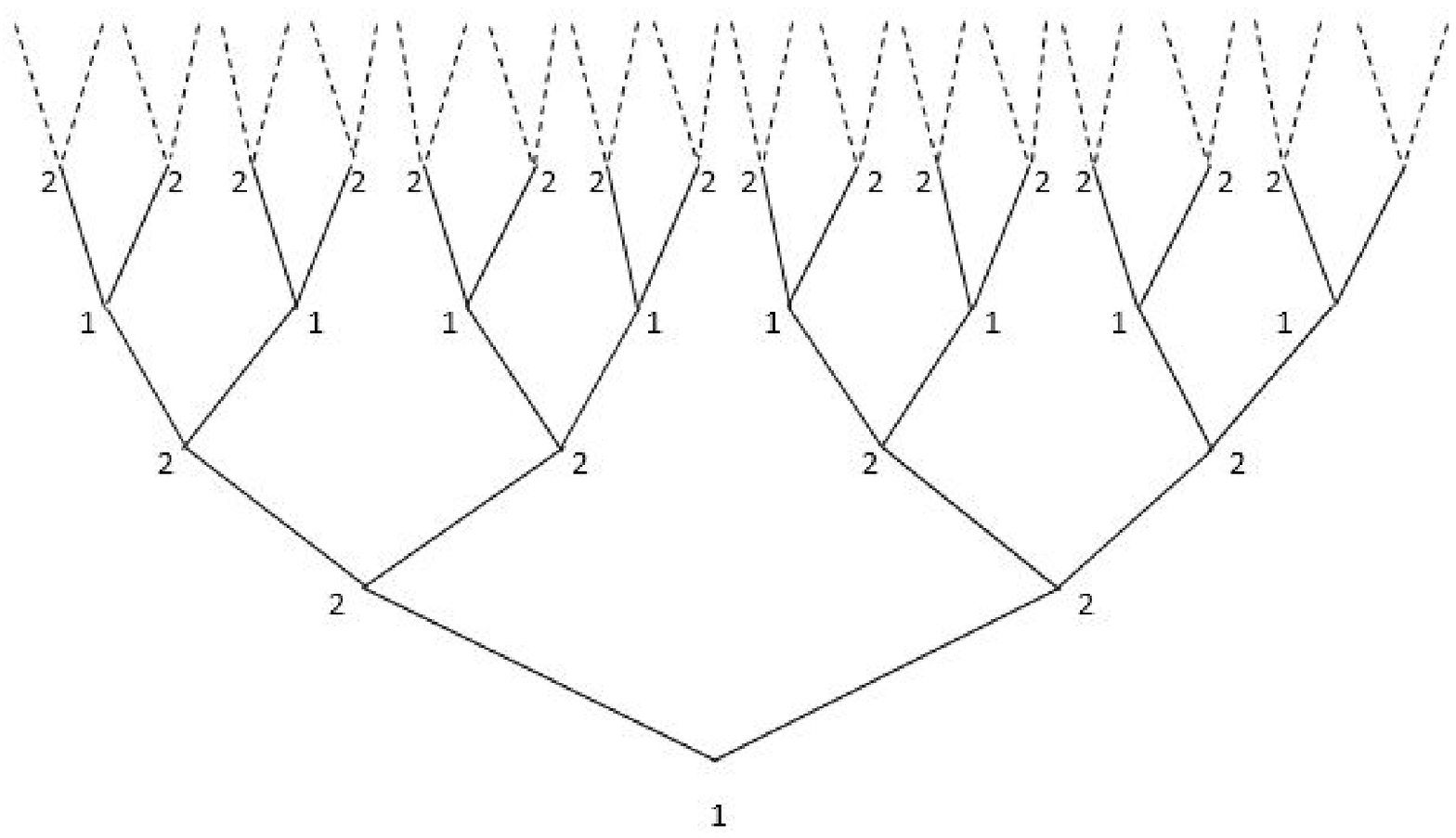}
            \end{center}
            \caption{Configuration for $\s^{(n)}$}
        \end{figure}

        Then, we can see that for any $b \in M$, one has
        \[U(\s^{(n)})=\min_{1 \leq k \leq 6}\{U_k\}\]
        which means $\s^{(n)}$ is a $G_n$-periodic ground state.

(iii) To construct uncountable number of ground states, we consider
the set
        \begin{eqnarray}\label{ncount2}
        \S_{2,3}=\left\{(\tb_n)| \tb_n \in \{2,3\},n \in \mathbb{N}\right\}
        \end{eqnarray}
        which is uncountable.
        Take any $\tb=(t_n)\in \S_{2,3}$. Let us construct a configuration by
        \begin{eqnarray*}
            \s^{(\tb)}=
            \left\{
            \begin{array}{ll}
                2,& x=(0), \\
                \tb_k,& x \in W_k,k \in V.
            \end{array}
            \right.
        \end{eqnarray*}
        \begin{figure}[h!]
            \begin{center}
                \includegraphics[width=13cm]{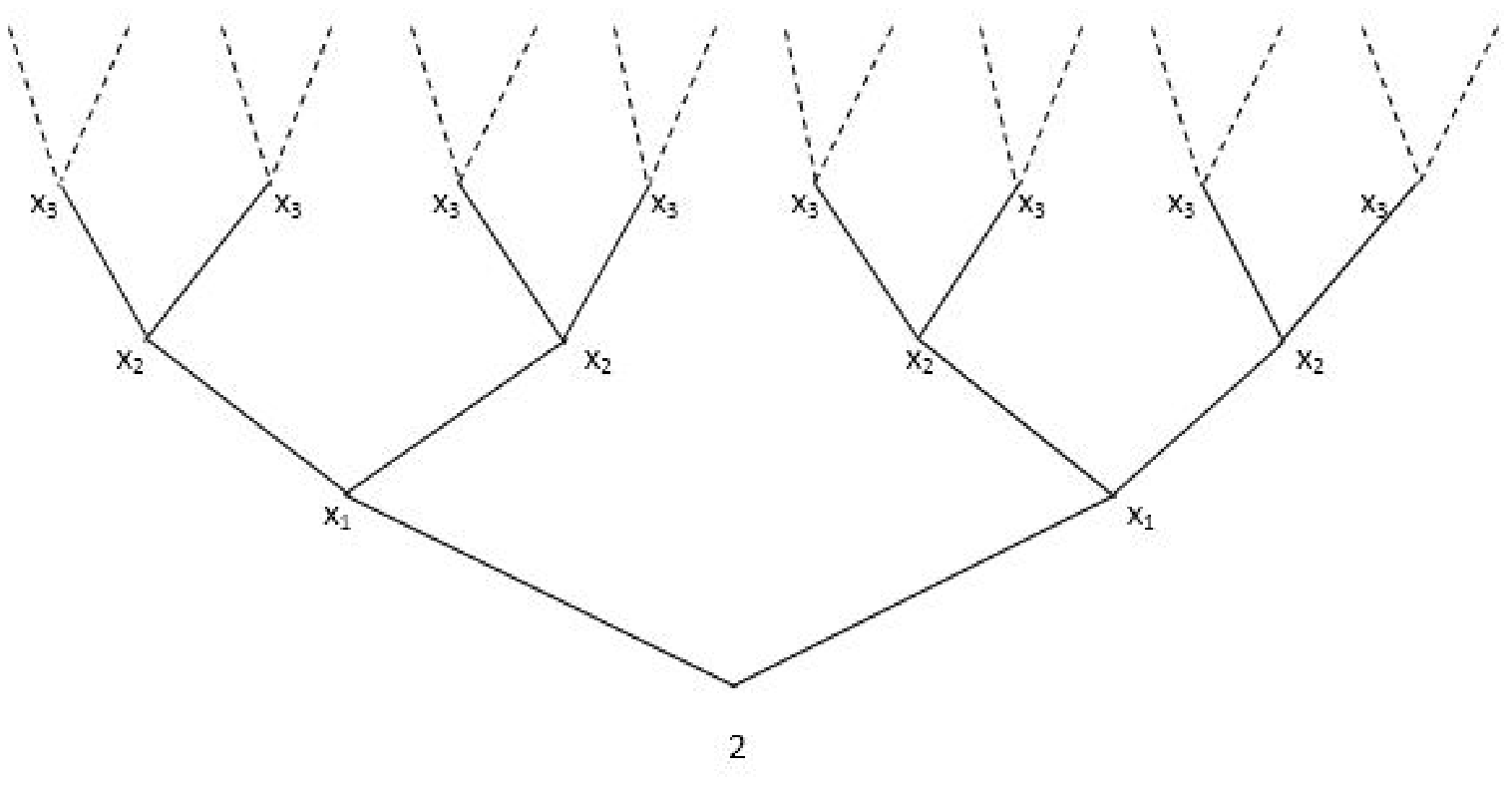}
            \end{center}
            \caption{Configuration for $\s^{(\tb)}$} \label{2,x1}
        \end{figure}
        One can check that $\s^{(\tb)}$ is a ground state, and the correspondence $\tb \in \S_{2,3} \rightarrow \s^{(t)}$ shows that the set $\{\s^{(\tb)},\tb\in \S_{2,3}\}$ is uncountable.
 \end{proof}

\begin{thm}
    Let $(\ab,\bb,\cb) \in A_6$, then there are only three transition-invariant ground states.
\end{thm}
\begin{proof}
    Let $(\ab,\bb,\cb) \in A_6$,then one can see that for this triple, the minimal value is $\cb$,
    which is achieved by the configurations on b given in Figure \ref{fig7}.
    \begin{figure}[h!]
        \begin{center}
            \includegraphics[width=8cm]{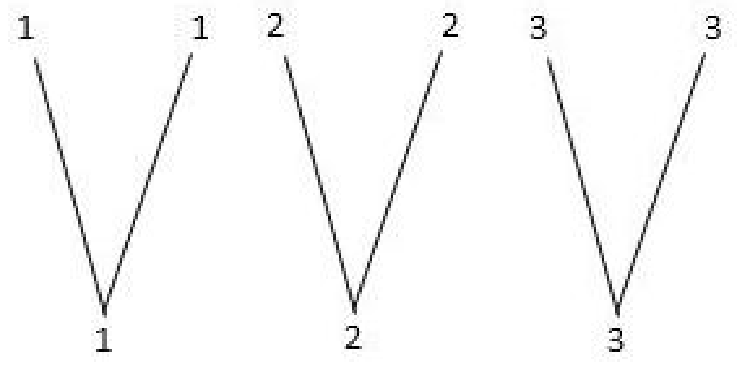}
        \end{center}
        \caption{Configurations for $A_6$} \label{fig7}
    \end{figure}

    In this case, we have three configurations:
    \[\s^{(k)}=\s_{[k]} ,k=\{1,2,3\} \]  which are translation-invariant ground states.
\end{proof}

\section{Construction of Gibbs States for the $\l$-model}

We define a finite-dimensional  distribution of probability measure
$\m^{(n)}$ in a volume $V_n$ as
\begin{eqnarray}\label{eq2}
\m^{(n)}(\s_n)=Z^{-1}_n{\exp\{\b H_n(\sigma_n)+\sum\limits_{x\in W_n} h_{\sigma(x),x}\}\ \ ,\s_n \in \F^{V_n}}
\end{eqnarray}
where $\b=1/T$, $T > 0$ is the temperature, and
\begin{eqnarray*}
    Z_n^{-1}=\sum_{\s\in\F^{V_n}}{\exp\{\b H_n(\sigma_n)+\sum\limits_{x\in W_n} h_{\sigma(x),x}\}}
\end{eqnarray*}
is the normalizing factor. In \eqref{eq2},
$\{h_x=(h_{1,x},\dots,h_{q,x}) \in \mathbb{R},x \in V \}$ is the set
of vectors, and
\begin{eqnarray*}
    H_n(\sigma_n)=\sum\limits_{<x,y>\in L_n}\l(\sigma (x),\sigma (y))
\end{eqnarray*}
We say that sequence of a probability distribution $\{\mu^{(n)}\}$ is consistent if for all $n \geq 1$ and $\s_{n-1}\in \F^{V_{n-1}}$ one has

\begin{eqnarray}\label{eq3}
\sum_{\w_n \in \F^{W_n}}\mu^{(n)}(\s_{n-1}\vee
\w_n)=\m^{(n-1)}(\s_{n-1})
\end{eqnarray}

Here,$\s_{n-1}\vee \w_n$ is the union of all configurations. In this case, we have a unique measure $\m$ on $\F^V$ such that for all $n$ and $\s_n \in \F^{V_n}$, we have
\begin{eqnarray*}
    \m(\{\s \vert V_n=\s_n\})=\m^{(n)}(\s_n)
\end{eqnarray*}
Such a measure is called a \textit{splitting Gibbs measure}
corresponding to Hamiltonian \eqref{ham} and to the vector-valued
function $h_x,x \in V$ (see \cite{Roz} for more information about
splitting measures).

The next statement describes the condition on $h_x$ ensuring that
the sequence $\{\m^{(n)}\}$ is consistent.

\begin{thm}\label{cons}
    The measures $\mu^{(n)},\ n=1,2,\dots ,$ satisfy the consistency condition  if and only if for any $x \in V$ the following equation holds:

    \begin{eqnarray}\label{eq4}
    u_{k,x}=\prod_{y \in S(x)}\dfrac{\sum_{j=1}^{q-1}\exp\{\b \l (k,j)\}u_{j,y}+\exp\{\b \l(k,q)\}}{\sum_{j=1}^{q-1}\exp\{\b \l (q,j)\}u_{j,y}+\exp\{\b \l(q,q)\}},
    \end{eqnarray}
    where $u_{k,x}=\exp\{h_{k,x}-h_{q,x}\}$,$k=\overline{1,q-1}$
\end{thm}

\begin{proof}[Proof: Necessity]
    According to the consistency condition \eqref{eq3}, we have
    \[\sum_{w \in \F^{V_n}} \dfrac{1}{Z_n}\exp\{\b H_n(\s_{n-1}\vee \w)+\sum_{x \in W_n}h_{\w(x),x}\}=\dfrac{1}{Z_{n-1}}\exp\{\b H_{n-1}(\s_{n-1})+\sum_{x \in W_{n-1}}h_{\s(x),x})\}\]
    Keeping in mind that
    \begin{align*}
    H_n(\s_{n-1}\vee \w_n)&=\sum_{\s \in \F^{V_{n-1}}}\l(\s (x),\s (y))+\sum_{\substack{x \in W{n-1}\\y \in S(x)}}\l(\s (x),\w (y))\\
    &=H_{n-1}(\s)+\sum_{x \in W_{n-1}}\sum_{y \in S(x)}\l(\s (x),\w (y)),
    \end{align*}

We have

    \begin{eqnarray*}
        \dfrac{Z_{n-1}}{Z_n}\sum_{\w \in \F^{V_n}}\exp\{\b\sum_{x \in W_{n-1}}\sum_{y \in S(x)}\l(\s (x),\w (y))+\b\sum_{x \in W_{n-1}}\sum_{y \in S(x)}h_{\w (x),(x)}=\exp\{\sum_{x \in W_{n-1}}h_{\s (x),(x)}\}
    \end{eqnarray*}
which yields
    \begin{eqnarray}\label{eq5}
    \dfrac{Z_{n-1}}{Z_n}\prod_{x \in W_{n-1}}\prod_{y \in S(x)}\sum_{\w \in \F^{V_n}}\exp\{\b \l(\s (x),\w (y))+h_{\w (y),(y)}=\prod_{x \in W_{n-1}}\exp\{h_{\s (x),(x)}\}
    \end{eqnarray}
    Considering configurations $\overline{\s}^{(k)} \in \F^{v_{n-1}}$,$\F=\{1,\dots,q\}$ such that $\s(x)=k$ for fixed $x \in V$ and
    $k=\overline{1,q}$, and dividing \eqref{eq5} at $\overline{\s}^{(k)}$ by \eqref{eq5} at $\overline{\s}^{(q)}$, one
    gets

    \begin{eqnarray}
    \prod_{y \in S(x)}\dfrac{\sum_{w \in \F}\exp\{\b\l(k,\w(y))+h_{\w(y),y}\}}{\sum_{\w \in \F}\exp\{\b\l(q,\w(y))+h_{\w(y),y}\}}=\dfrac{\exp\{h_{k,x}\}}{\exp\{h_{q,x}\}}.
    \end{eqnarray}
    So,
    \begin{eqnarray}\label{eq6}
    \prod_{y \in S(x)}\dfrac{\sum_{w \in \F}^{q}\exp\{\b\l(k,j)+h_{j,y}\}}{\sum_{w \in \F}^{q}\exp\{\b\l(q,j)+h_{j,y}\}}=\exp\{h_{k,x}-h_{q,x}\}.
    \end{eqnarray}

    Hence, by denoting $u_{k,x}=\exp\{h_{k,x}-h_{q,x}\}$ from \eqref{eq6} one finds

    \begin{eqnarray}\label{eq7}
    \prod_{y \in S(x)}\dfrac{\sum_{j=1}^{q-1}\exp\{\b\l(k,j)+u_{j,y}\}+\exp\{\b\l(k,q)\}}{\sum_{j=1}^{q-1}\exp\{\b\l(q,j)+u_{j,y}\}+\exp\{\b\l(q,q)\}\}}=u_{k,x}
    \end{eqnarray}
\end{proof}

\begin{proof}[Sufficiency]
    Suppose that \eqref{eq3} holds, then we get \eqref{eq6}. which yields that

    \begin{eqnarray}\label{eq8}
    \prod_{y \in S(x)}\sum_{j=1}^{q}\exp\{\b\l(k,j)+h_{j,y}\}=a(x)\exp\{h_{k,x}\},\ k=\overline{1,q},\ x \in W_{n-1}
    \end{eqnarray}
    for some function $a(x)>0$, $x \in V$.

    Let us multiply \eqref{eq8} with respect to $x \in W_{n-1}$, then we obtain

    \begin{eqnarray}\label{eq9}
    \prod_{x \in W_{n-1}}\prod_{y \in S(x)}\sum_{j=1}^{q}\exp\{\b\l(\s(x),j)+h_{j,y}\}=\prod_{x \in W_{n-1}}(a(x)\exp\{h_{k,x}\})
    \end{eqnarray}
    for any configuration $\s \in \F^V_{n-1}$. Denoting $A_{n-1}=\prod_{x \in W_{n-1}}a(x)$, from \eqref{eq9}, one finds

    \begin{eqnarray}\label{eq10}
    \prod_{x \in W_{n-1}}\prod_{y \in S(x)}\sum_{w \in \F}\exp\{\b\l(\s(x),\w(y))+h_{\w,y}\}=A_n\prod_{x \in W_{n-1}}\exp\{h_{\s(x),x}\}
    \end{eqnarray}
    We multiply both sides of \eqref{eq10} by $\exp\{\b H_{n-1}(\s)\}$, we get

    \begin{eqnarray*}
        \exp\{\b H_{n-1}(\s)\}\prod_{x \in W_{n-1}}\prod_{y \in S(x)}\sum_{\w \in \F}\exp\{\b\l(\s(x),\w(y))+h_{\w(y),y}\}=A_n\exp\{\b H_{n-1}(\s)\}\prod_{x \in W_{n-1}}\exp\{h_{\s(x),x}\}
    \end{eqnarray*}
    which yields

    \begin{eqnarray}\label{eq11}
    Z_n\sum_{\w_{n} \in \F^{V_n}}\m^{(n)}(\s \vee \w_n)=A_{n-1}Z_{n-1}\m^{(n-1)}(\s)
    \end{eqnarray}
    since $\m^{(n)}(\s \vee \w_n)$, $n \geq 1$ is a probabilistic measure, we have

    \begin{eqnarray*}
        \sum_{\w_n \in \F^{V_n}}\m^{(n)}(\s \vee \w_n)=\m^{(n-1)}(\s)=1
    \end{eqnarray*}
    which from \eqref{eq11} yields
    \[Z_n=A_{n-1}Z_{n-1}\]
    This completes the proof.
\end{proof}

\section{Description of translation-invariant Gibbs measures.}

In this section, we are going to establish the existence of phase
transition for the $\l$-model give by \eqref{cond}. As before, in
what follows, we assume that $k=2$, $q=3$.

To establish a phase transition, we will find translation-invariant
Gibbs measures. Here, by translation-invariant Gibbs measure we mean
a splitting Gibbs measure which correspond to a solution
$\mathbf{u}_x$ of the equation \eqref{eq7} which is
translation-invariant, i.e. $\mathbf{u}_x=\mathbf{u}_y$ for all $x,y
\in V$. This means $\mathbf{u}_x=\mathbf{u}$, where
$\mathbf{u}=(u_1,u_2),\ u_1,u_2 > 0$. Due to Theorem \ref{cons},
$u_1$ and $u_2$ must satisfy the following equation:
\begin{eqnarray}\label{lam}
u_1=\left(\dfrac{{u_1}\x+u_2\y+\z}{{u_1}\z+u_2\y+\x}\right)^2,
\ u_2=\left(\dfrac{{u_1}\y+u_2\x+\y}{{u_1}\z+u_2\y+\x}\right)^2,
\end{eqnarray}
where $\x=\exp\{\b\cb\},\ \y=\exp\{\b\bb\},\ \z=\exp\{\b\ab\}$ (here
we have used \eqref{cond}).

From \eqref{lam} one can see that $u_1=1$ is invariant line of the equation. Therefore, the equation over this invariant line reduces to
\begin{eqnarray}\label{f}
u_2=\left(\dfrac{\x{u_2}+2\y}{\y{u_2}+\x+\z}\right)^2.
\end{eqnarray}
Denoting
\begin{eqnarray}\label{value}
x=\dfrac{u_2\x}{2\y},&a=\dfrac{1}{8\y^5},&b=\dfrac{\x(\x+\z)}{2\y^2},
\end{eqnarray}
we rewrite \eqref{f} as follows:
\begin{eqnarray}\label{gen}
ax=\left(\dfrac{1+x}{b+x}\right)^2.
\end{eqnarray}
Since $x > 0,k \geq 1,a > 0,$ and $b > 0$;\cite[Proposition
10.7]{Pr} implies the following
\begin{lem}\label{lem1}
    \begin{itemize}
        \item[(1).] If $\ b \leq 9$ then a solution to \eqref{gen} is unique.

        \item[(2).] If $\ b > 9$ then there are $\e_1(b)$ and $\e_2(b)$ such that $0 < \e_2 < \e_2$ and if $\e_1 < a < \e_2$ then \eqref{gen} has three solutions.

        \item[(3).] If $a=\e_1$ and $a=\e_2$ then \eqref{gen} has two solutions. The quantity $\e_1$ and $\e_2$ are determined from the formula
        \begin{eqnarray}\label{eta}
        \e_i(b)=\dfrac{1}{x_i}\left(\dfrac{1+x_i}{b+x_i}\, \right)^2, \ i=1,2,
        \end{eqnarray}
        where $x_1$ and $x_2$ are solutions to the equation $x^2+(3-b)x+b=0$.
    \end{itemize}
\end{lem}
From Lemma \ref{lem1} we conclude the following result:

\begin{thm}\label{phase}
    If condition (2) of lemma \ref{lem1} is satisfied, then there occurs a phase transition.
\end{thm}

Let us consider some concrete examples.

\begin{ex}
    Let $b=10$. We have $x_1=2$ and $x_2=5$. Then we have $\e_1=1/32$ and $\e_2=4/125$. So from theorem \ref{phase}, we can conclude that if $\dfrac{1}{32}<\dfrac{2\y^3}{\x^3}<\dfrac{4}{125}$, then there occurs a phase transition.
\end{ex}

\section{Periodic Gibbs Measure}

In this section, we are going to study 2-periodic Gibbs measures. Recall that function $\ub_x$ is 2-periodic if $\ub_x=\ub_y$ whenever $d(x,y)$ is divisible by 2 (see for detail section (2)).

Let $\ub_x$ be a 2-periodic function. Then, to exist the
corresponding Gibbs measure, the function $\ub_x=(u_{x,1},u_{x,2})$
should satisfy the following equation:
\begin{eqnarray}\label{ux}
u_{x,1}=\left(\dfrac{{u_{y,1}}\x+u_{y,2}\y+\z}{{u_{y,1}}\z+u_{y,2}\y+\x}\right)^2,
\ u_{x,2}=\left(\dfrac{{u_{y,1}}\y+u_{y,2}\x+\y}{{u_{y,1}}\z+u_{y,2}\y+\x}\right)^2,
\end{eqnarray}
where $d(x,y)=2$ for all $x,y \in V$.

According to the previous section,$u_{x,1}=1$ is invariant line for
the equation \eqref{ux}. Therefore, in what follows, we assume
$u_{x,1}=1$ for all $x \in V$. Then, \eqref{ux} reduces to

\begin{eqnarray}\label{period}
u=f(f(u))
\end{eqnarray}
where
\begin{eqnarray*}
    f(u)=\left(\dfrac{\x{u}+2\y}{\y{u}+\x+\z}\right)^2.
\end{eqnarray*}
Roots of $u_2=f(u)$ are clearly roots of Eq.\eqref{period}. In order
to find the other roots of Eq. \eqref{period} that differ from the
roots of $u=f(u)$, we must therefore consider the equation
\begin{eqnarray*}
    \dfrac{f(f(u))-u}{f(u)-u}=0,
\end{eqnarray*}
which yields the quadratic equation
\begin{eqnarray}\label{bal}
\begin{aligned}
(\y^2\z^2+2\y\x^2\z+2\z\y^2\x+\x^4+2\y\x^3+\x^2\y^2)u^2+(\x^2\z^2+6\y\x^3+8\x^2\y^2+8\z\y^2\x\\
+2\y\z^3+2\x^3\y+6\y\x^2\z-4\y^4+6\y\z^2\x+\x^4)u+\x^4+4\x^2\y^2+4\y\x^3\\
+\z^4+4\x^3\z+4\y\z^2\x+6\x^2\z^2+8\y\x^2\z+4\x\z^3=0
\end{aligned}
\end{eqnarray}
Note that the positive roots of \eqref{bal} generate periodic Gibbs measures. In general, the existence of two positive roots are given by the following conditions:
\begin{eqnarray}\label{BD}
B <\ 0,\ D >\ 0,
\end{eqnarray}
where
\begin{align*}
B=&\x^2\z^2+6\y\x^3+8\x^2\y^2+8\z\y^2\x+2\y\z^3+2\x^3\z+6\y\x^2\z-4\y^4+6\y\z^2\x+\x^4\\
D=&16\y^8-3\x^8+96\y^3\x^4\z-64\z\y^6\x+40\x^2\z^2\y^4-40\x^4\z^3\y-48\y^5\x^2\z-48\y^5\z^2\x-20\x^3\z^4\y\\
&-48\x^5\z^2\y+64\x^2\y^3\z^3-4\x^2\z^5\y+16\z^4\y^3\x-36\y\x^6\z+96\x^3\y^3\z^2+80\x^3\y^4\z-3\x^4\z^4-12\x^5\z^3\\
&-18\x^6\z^2+48\y^3\x^5-48\y^5\x^3-12\y\x^7+40\x^4\y^4-64\x^2\y^6-16\y^5\z^3-12\x^7\z.
\end{align*}
Let us consider several cases.
\begin{itemize}
    \item[(i)] Let $\z=\x+1$ and $\y=1$, then we have
    \begin{eqnarray*}
        B&=&4\x^4+24\x^3+41\x^2+20\x-2,\\
        D&=&-144\x+344\x^3-156\x^2+541\x^4-48\x^8-256\x^7-424\x^6-32\x^5.
    \end{eqnarray*}
    We can factor $D$ as follows,
    \begin{eqnarray*}
        D=-\x(4+3\x)(2\x+3)^2(2\x^2+\x-2)^2.
    \end{eqnarray*}
    Then $D < 0$, i.e., all the 2-periodic Gibbs measures are translation invariant.

    \item[(ii)] Let $\z=\x$ and $\y=1$, then one has
    \begin{eqnarray*}
        B&=&4\x^4+20\x^3+16\x^2-4,\\
        D&=&-48\x^8-160\x^7+160\x^4+320\x^5-160\x^3-128\x^2+16.
    \end{eqnarray*}
    We can factor $D$ as follow
    \begin{eqnarray*}
        D=-16(3\x^4+10\x^3+6\x^2-1)(\x-1)^2(\x+1)^2.
    \end{eqnarray*}
    Using a MAPLE program, we find that the equation $D=0$ has two real roots, such that one of them is positive,
    \begin{eqnarray*}
        \theta_D \approx 0.323591553488076
    \end{eqnarray*}
    Hence, if $\x > \theta_D $, then $D < 0$, i.e., all the 2-periodic Gibbs measures are translation invariant. If $0 < \x < \theta_D$, then $D > 0$ and $B < 0$ (see Figure 17) which implies the existence of 2-periodic Gibbs measure which implies the following result:

    \begin{thm}
        Let $\z=\x$ and $\y=1$. If $0 < \x < \theta_D$, then there exist a phase transitions.
    \end{thm}

    \begin{rem}
        We note that under condition theorem 6.1, one can find translation-invariant Gibbs measure if $\x > 3$. Theorem 6.1 means that the existence of 2-periodic Gibbs measure does not implies the existence of translation-invariant Gibbs measure.
    \end{rem}
    \begin{figure}[h!]
        \begin{center}
            \includegraphics[width=10cm]{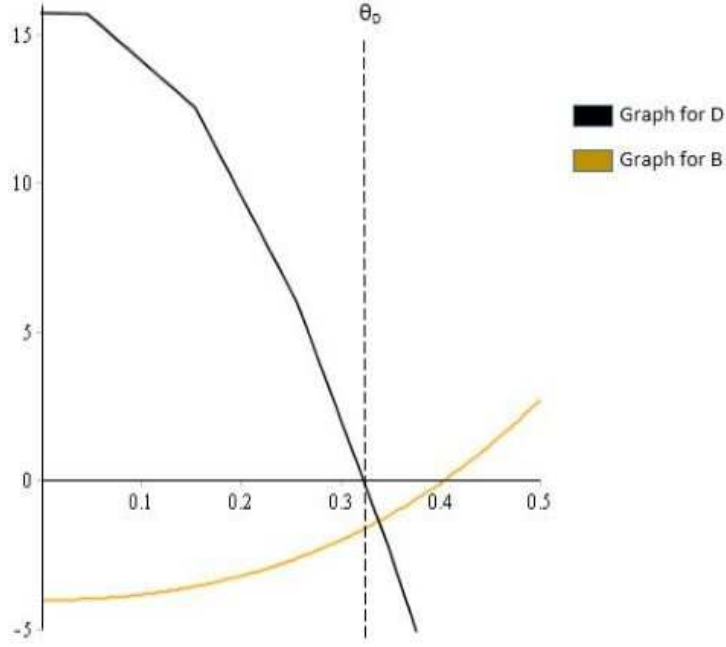}
        \end{center}
        \caption{Existence of 2-periodic solution (case (ii))} \label{case 2}
    \end{figure}

    \item[(iii)] Let $\x=\y$ and $\z$ is arbitrary number, then one has
    \begin{eqnarray*}
        B&=&7\x^2\z^2+11\x^4+16\z\x^3+2\x\z^3,\\
        D&=&-7\x^4\z^4+22\x^6\z^2-4\x^5\z^3-4\x^3\z^5-23\x^8+16\z^7\z.
    \end{eqnarray*}
    We can factor $D$ as follows
    \begin{eqnarray*}
        D=-\x^3(23\x^3+30\x^2\z+15\x\z^2+4\z^3)(-\z+\x)^2.
    \end{eqnarray*}
    Hence, for any value of $\x,\z > 0$ then $D < 0$, i.e., all the 2-periodic Gibbs measures has translation-invariant.

\end{itemize}

\smallskip

\end{document}